\documentclass[11pt,a4paper]{article}

\usepackage[utf8]{inputenc}
\usepackage{amsfonts}
\usepackage{mathtools}
\usepackage{amsthm}
\usepackage{amssymb}
\usepackage{lmodern}
\usepackage{tabularx,environ}
\usepackage{hyperref}
\usepackage[linesnumbered,lined,boxed,commentsnumbered,noend]{algorithm2e}

\usepackage{graphicx}
\usepackage{indentfirst}
\usepackage[margin=2.9cm]{geometry}
\usepackage{caption,subcaption}

\usepackage{authblk}

\usepackage{tikz}
\usetikzlibrary{shapes, arrows.meta, automata, positioning, fit, shadows, decorations.pathmorphing, backgrounds}
\usetikzlibrary{decorations.pathreplacing}
\usetikzlibrary{calc}
\usetikzlibrary{shapes.geometric}

\input{problem_box.sty}

\theoremstyle{plain}
\newtheorem{theorem}{Theorem}[section]
\newtheorem{lemma}[theorem]{Lemma}
\newtheorem{corollary}[theorem]{Corollary}
\newtheorem{observation}[theorem]{Observation}

\theoremstyle{definition}
\newtheorem{definition}[theorem]{Definition}

\DeclareMathOperator{\reach}{R}
\DeclareMathOperator{\wreach}{W}
\DeclareMathOperator{\bcon}{B}
\DeclareMathOperator{\col}{col}
\DeclareMathOperator{\wcol}{wcol}
\DeclareMathOperator{\adm}{adm}
\DeclareMathOperator{\degen}{degen}
\DeclareMathOperator{\depth}{depth}
\DeclareMathOperator{\est}{est}

\mathchardef\hyph="2D

\newcommand{\tcts}{\probname{2-Clause 3-SAT}}
\newcommand{\rsat}{\probname{Exact $r$-SAT}}

\newcommand{\affilformat}[6]{\affil[#1]{\textbf{#2}\vskip 0pt \textbf{#3}\vskip 0pt #4\vskip 0pt #5\vskip 0pt #6}}

\title{Hardness of the Generalized Coloring Numbers}
\author[1]{Michael Breen-McKay}
\affilformat{1}{North Carolina State University}{Department of Mathematics}{2108 SAS Hall}{Raleigh, NC 27695}{mikebreenmckay@gmail.com\linebreak}

\author[2]{Brian Lavallee}
\author[2]{Blair D.\ Sullivan}
\affilformat{2}{University of Utah}{School of Computing}{50 Central Campus Drive}{Salt Lake City, UT 84112}{brian.lavallee@utah.edu\ \ \ sullivan@cs.utah.edu}

\date{15 September 2022}

\begin{document}

\maketitle

\begin{abstract}
	The generalized coloring numbers of Kierstead and Yang (Order 2003) offer an algorithmically-useful characterization of graph classes with bounded expansion.
	In this work, we consider the hardness and approximability of these parameters.
	First, we complete the work of Grohe et al.\ (WG 2015) by showing that computing the weak 2-coloring number is NP-hard.
	Our approach further establishes that determining if a graph has weak $r$-coloring number at most $k$ is para-NP-hard when parameterized by $k$ for all $r \geq 2$.
	We adapt this to determining if a graph has $r$-coloring number at most $k$ as well, proving para-NP-hardness for all $r \geq 2$.
	Para-NP-hardness implies that no XP algorithm (runtime $O(n^{f(k)})$) exists for testing if a generalized coloring number is at most $k$.
	Moreover, there exists a constant $c$ such that it is NP-hard to approximate the generalized coloring numbers within a factor of $c$.
	To complement these results, we give an approximation algorithm for the generalized coloring numbers, improving both the runtime and approximation factor of the existing approach of Dvo\v{r}\'{a}k (EuJC 2013).
	We prove that greedily ordering vertices with small estimated backconnectivity achieves a $(k-1)^{r-1}$-approximation for the $r$-coloring number and an $O(k^{r-1})$-approximation for the weak $r$-coloring number.
\end{abstract}

\addtocounter{footnote}{1}
\footnotetext{This work was supported in part by the Gordon \& Betty Moore Foundation under award GBMF4560 to Blair D.\ Sullivan.}

\section{Introduction}
In 2008, Ne\v{s}et\v{r}il and Ossona de Mendez introduced \emph{bounded expansion} as a notion of sparsity in graph classes~\cite{nesetril2008grad}.
Bounded expansion generalizes \emph{proper minor-closed} classes which contain (among other things) the heavily-studied classes of \emph{bounded treewidth} and \emph{planar} graphs.
Ne\v{s}et\v{r}il and Ossona de Mendez characterized classes with bounded expansion by the existence of low-treedepth colorings~\cite{nesetril2008grad}, enabling the first of many algorithmic approaches to tackle problems on these broad classes.
While part of the appeal of bounded expansion is its numerous characterizations, low-treedepth colorings alone allowed Dvo\v{r}\'{a}k et al.\ to prove that \probname{FO-Model Checking} is \emph{fixed-parameter tractable} (FPT) in graphs from a class of bounded expansion~\cite{dvorak2010deciding}.
Beyond the many problems expressible in first order logic, characterizations of bounded expansion have been used to give FPT algorithms for \probname{Subgraph Counting}~\cite{nesetril2008GradAC}, \probname{Steiner Tree}~\cite{kreutzer2019algorithmic}, and numerous others~\cite{drange2016kernelization, drange2021harmless, harpeled2015approximation}.

In this work, we focus on the \emph{generalized coloring numbers}, introduced by Kierstead and Yang in 2003~\cite{kierstead2003orderings}, which also characterize classes of bounded expansion~\cite{zhu2009coloring}.
The generalized coloring numbers (consisting of the \emph{$r$-coloring numbers} and the \emph{weak $r$-coloring numbers}) extend the notion of \emph{degeneracy} to neighborhoods of radius $r$.
The generalized coloring numbers are also deeply connected to other important graph properties.
In fact, the $n$-coloring number and the weak $n$-coloring number are equivalent to the \emph{treewidth} and \emph{treedepth} of a graph with $n$ vertices, respectively.
Prior to the introduction of bounded expansion, the generalized coloring numbers were studied primarily in the context of the \emph{game chromatic number}~\cite{kierstead2003orderings}.\looseness-1

Zhu's connection of these parameters to bounded expansion in 2009~\cite{zhu2009coloring} renewed interest in analyzing and utilizing them.
Reidl et al.\ used the weak coloring numbers to prove that \emph{neighborhood complexity} characterizes bounded expansion~\cite{reidl2019neighborhood}, while Nadara et al.\ initiated an empirical study of algorithmic techniques used to compute the generalized coloring numbers in practice~\cite{nadara2019empirical}.
The generalized coloring numbers also led to new algorithms for bounded expansion;
Dvo\v{r}\'{a}k gave an approximation for \probname{Dominating Set}~\cite{dvorak2013constant}, and Reidl and Sullivan gave an exact FPT algorithm for \probname{Subgraph Counting}~\cite{reidl2023color}.
Both of these algorithms rely on computing a linear order of the vertices which witnesses a small weak $r$-coloring number for some fixed radius $r$.
Grohe et al.\ showed that determining the weak $r$-coloring number is NP-hard when $r \geq 3$~\cite{grohe2018coloring}, and so both algorithms use an approximation to find a suitable order.
Dvo\v{r}\'{a}k gave a parameterized approximation for the weak $r$-coloring number in $O(rn^3)$ time through an approximation for \emph{$r$-admissibility}~\cite{dvorak2013constant}.

In this paper, we prove that both of the generalized coloring numbers are para-NP-hard to compute when $r \geq 2$.
For the weak coloring numbers, this improves upon the result of Grohe et al.~\cite{grohe2018coloring} who showed that the weak $r$-coloring number is NP-hard to compute when $r \geq 3$.
For the coloring numbers, no hardness results were previously known.
Our approach uses \tcts in contrast to \probname{Balanced Complete Bipartite Subgraph} as in~\cite{grohe2018coloring}.
Para-NP-hardness implies that there does not exist an XP algorithm (i.e.\ an algorithm with runtime $O(n^{f(k)})$) to check if a generalized coloring number is at most $k$.
Additionally, it implies that there exists a constant $c$ such that it is NP-hard to approximate the generalized coloring numbers within a factor of $c$.

Finally, we show that there is a $(k-1)^{r-1}$-approximation algorithm for the $r$-coloring number, improving both the runtime and approximation guarantees given by Dvo\v{r}\'{a}k~\cite{dvorak2013constant}.
The algorithm applies a greedy strategy directly to the coloring number rather than using admissibility as an intermediary.
We further prove that the order produced by this algorithm gives an $O(k^{r-1})$-approximation for the weak $r$-coloring number.

\section{Preliminaries}
We examine properties of \emph{total orders} defined on the vertices of simple, undirected graphs.
For a graph $G = (V, E)$, we will frequently use $n = |V|$ and $m = |E|$ to denote the number of vertices and edges respectively.
For flexibility in our analysis, we define the properties of interest using \emph{prefix orders}.

\begin{definition} \label{def:prefix_order}
    A \textbf{prefix order} is a \emph{strict partial order} which is also \emph{downward total}.
    That is, a prefix order $\sigma$ of a set $V$ satisfies the following conditions for all $a, b, c \in V$:
    \begin{itemize}
        \setlength\itemsep{-2.6pt}

        \item
        irreflexivity: $a \not<_\sigma a$

        \item
        asymmetry: if $a <_\sigma b$, then $b \not<_\sigma a$

        \item
        transitivity: if $a <_\sigma b$ and $b <_\sigma c$, then $a <_\sigma c$

        \item
        downward totality: if $a <_\sigma c$ and $b <_\sigma c$, then either $a <_\sigma b$ or $b <_\sigma a$
    \end{itemize}
    We say an element $a$ is \textbf{ordered} if there exists an element $b$ such that $a <_\sigma b$.
    All other elements are \textbf{unordered}.
\end{definition}

Intuitively, a prefix order is an ``incomplete'' total order where the position of only the first $i$ elements has been decided and the remaining unordered elements will be added to the end of the order.
Given a prefix order $\sigma$ of the vertices of $G$, we use $G_\sigma$ to refer to the ordered graph defined by the combination of $G$ and $\sigma$.
A prefix order naturally partitions the neighborhood of a vertex $u$ into its forward and backward neighbors.
The same method of partitioning can also be applied to the $r$-neighborhood of $u$ (i.e.\ the set of vertices of distance at most $r$ from $u$ in $G$), but more restrictive notions of forward neighbors have better theoretical and algorithmic applications.
Specifically, we use $r$-reachable and weak $r$-reachable as criteria to determine if a vertex $v$ is a forward neighbor of $u$.

\begin{definition} \label{def:reachable}
    Given a graph $G = (V, E)$ and a prefix order $\sigma$, let $P \subseteq G$ be a $u$-$v$ path of length at most $r$.
    If $v \not<_\sigma u$ and $p <_\sigma u$ for all $p \in P \setminus \{u,v\}$, then we say that $P$ is an \textbf{$r$-qualifying path} and $v$ is \textbf{$r$-reachable} from $u$.
    Similarly, if $v \not<_\sigma u$ and $p <_\sigma v$ for all $p \in P \setminus \{u,v\}$, then we say $P$ is a \textbf{weak $r$-qualifying path} and $v$ is \textbf{weakly $r$-reachable} from $u$.
    Note that all $r$-qualifying paths are also weak $r$-qualifying paths.
\end{definition}

\begin{figure}
\begin{minipage}{0.48\textwidth}
    \centering
    \scalebox{0.91}{
    \begin{tikzpicture}
        \tikzstyle{leafVert}=[circle,minimum size=5pt,inner sep=0pt]

        \node[draw, leafVert] (B) {};
        \node[draw, right=.5cm of B, leafVert] (C) {};
        \node[draw, right=.5cm of C, leafVert] (D) {};
        \node[draw, right=.5cm of D, leafVert] (E) {};
        \node[draw, right=.5cm of E, leafVert, fill, blue] (F) {};
        \node[draw, right=.5cm of F, leafVert] (G) {};
        \node[draw, right=.5cm of G, leafVert] (H) {};
        \node[draw, right=.5cm of H, leafVert] (I) {};
        \node[draw, right=.5cm of I, leafVert] (J) {};
        \node[draw, right=.5cm of J, leafVert] (K) {};
        \node[draw, right=.5cm of K, leafVert] (L) {};

        \node[below=.2cm of B] (B1) {$v_0$};
        \node[below=.2cm of C] (C1) {$v_1$};
        \node[below=.2cm of D] (D1) {$v_2$};
        \node[below=.2cm of E] (E1) {$v_3$};
        \node[below=.2cm of F] (F1) {$v_4$};
        \node[below=.2cm of G] (G1) {$v_5$};
        \node[below=.2cm of H] (H1) {$v_6$};
        \node[below=.2cm of I] (I1) {$v_7$};
        \node[below=.2cm of J] (J1) {$v_8$};
        \node[below=.2cm of K] (K1) {$v_9$};
        \node[below=.2cm of L] (L1) {$v_{10}$};

        \draw  (C) -- (D) -- (E) -- (F);

        \draw (G)-- (H) -- (I) -- (J) -- (K) -- (L);

        \draw (F) -- (G);
        \draw (B) -- (C);

        \draw (B) to[out=45,in=135] (F);
        \draw (B) to[out=45,in=135] (I);
        \draw (B) to[out=45,in=135] (K);

        \draw (C) to[out=35,in=145] (H);

        \draw (F) to [out=35,in=145] (J);
    \end{tikzpicture}
    }
\end{minipage}
\begin{minipage}{0.51\textwidth}
    \[
    \begin{array}{c|c|c|c}
        r & \reach_r(v_4, G\sigma) & \wreach_r(v_4,G_\sigma) \\
        \hline
        1 &\{v_5,v_8\} & \{v_5,v_8\} \\
        2 & \{v_5,v_7,v_8,v_9\} & \{v_5,v_6,v_7,v_8,v_9\} \\
        3 & \{v_5,v_6,v_7,v_8,v_9\} & \{v_5,v_6,v_7,v_8,v_9,v_{10}\}
    \end{array}
    \]
\end{minipage}

\caption{
    We give a total order on the vertices of a graph $G$.
    The chart gives the various vertices reachable from $v_4$ under each definition.
}
\label{fig:linorderex}
\end{figure}
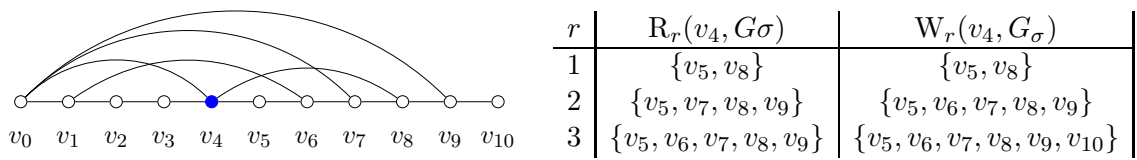

We note that there are minor inconsistencies in the definitions of these terms in prior work.
In particular, we do not consider a vertex to be reachable from itself (resulting in a value that is off-by-one from the definition in e.g.~\cite{grohe2018coloring}).
More significantly, we also reverse the direction of the order used in~\cite{grohe2018coloring} to adhere to a left-to-right convention.
Thus, a vertex reaches to the right using vertices to the left, rather than reaching to the left using vertices to the right.
Both definitions are equivalent, and this can be seen by simply reversing $\sigma$.

We use $\reach_r(u, G_\sigma)$ to denote the set of vertices in $G$ which are $r$-reachable from $u$ with respect to $\sigma$ and $\wreach_r(u, G_\sigma)$ to denote the set of vertices which are weakly $r$-reachable.
The \textbf{$r$-reach} and the \textbf{weak $r$-reach} refer to the sizes of these sets respectively.
Figure~\ref{fig:linorderex} shows the differences between $\reach_r(u, G_\sigma)$ and $\wreach_r(u, G_\sigma)$ on a small example graph.
Using the notions of $r$-reach and weak $r$-reach as analogues of forward degree leads to two natural generalizations of degeneracy.

\begin{definition} \label{def:colnum}
    Given a graph $G = (V, E)$, let $\Pi$ be the set of all total orders of $V$.
    The \textbf{$r$-coloring number} and the \textbf{weak $r$-coloring number} of a graph $G$ are defined as
    \[ \col_r(G) := \min\limits_{\sigma \in \Pi} \max\limits_{u \in V} |\reach_r(u, G_\sigma)|
    \hspace{1cm}
    \wcol_r(G) := \min\limits_{\sigma \in \Pi} \max\limits_{u \in V} |\wreach_r(u, G_\sigma)|. \]
\end{definition}

We will also use $\col_r(G_\sigma)$ and $\wcol_r(G_\sigma)$ to refer to the maximum $r$-reach and maximum weak $r$-reach over all vertices with respect to an order $\sigma$.
Under the given definitions, both of the generalized coloring numbers are closely linked with other important graph measures.
Specifically, $\col_n(G)$ is equal to the treewidth~\cite{grohe2018coloring} and $\wcol_n(G)$ is one less than the treedepth~\cite{nesetril2012sparsity}.
Since $\col_1(G) = \wcol_1(G) = \degen(G)$, these parameters offer an interpolation between degeneracy and treewidth or treedepth.
In this paper, we consider the problem of computing the generalized coloring numbers.

\begin{problem}{$r$-Orderable}
	\Input & a graph $G = (V, E)$ and an integer $k \in \mathbb{N}$. \\
	\Prob & is $\col_r(G) \leq k$?
\end{problem}

\begin{problem}{Weak $r$-Orderable}
	\Input & a graph $G = (V, E)$ and an integer $k \in \mathbb{N}$. \\
	\Prob & is $\wcol_r(G) \leq k$?
\end{problem}

We use \probname{Minimum $r$-Orderability} and \probname{Minimum Weak $r$-Orderability} to refer to the corresponding optimization problems.
Due to the equivalence of the 1-coloring number and weak 1-coloring number with degeneracy, we note that both \probname{1-Orderable} and \probname{Weak 1-Orderable} can be solved in linear time using a greedy algorithm.
Finally, our results make use of \rsat and \tcts, defined below.

\begin{problem}{\rsat}
    \Input & a CNF-formula $\phi$ with clauses $c_1, \dots, c_m$ consisting of variables $x_1, \dots, x_n$ such that each clause contains \emph{exactly} $r$ different variables. \\
    \Prob & is there an assignment of $x_1, \dots, x_n$ that satisfies $\phi$?
\end{problem}

\begin{problem}{\tcts}
	\Input & a CNF-formula $\phi$ with clauses $c_1, \dots, c_m$ consisting of variables $x_1, \dots, x_n$ such that each clause contains \emph{at most} 3 variables and each literal ($x_j$ or $\overline{x}_j$) appears in \emph{exactly} 2 clauses. \\
	\Prob & is there an assignment of $x_1, \dots, x_n$ that satisfies $\phi$?
\end{problem}

Karp proved that \probname{3-SAT} is NP-complete in~\cite{karp1972reducibility}, implying that \rsat is NP-hard for $r \geq 3$.
To see that \tcts is also NP-hard, we use Theorem 2.1 of~\cite{tovey1984simplified} which states that \probname{SAT} remains NP-complete when restricted to instances with 2 or 3 variables per clause and at most 3 occurrences per variable.
These instances of \probname{SAT} can be converted to \tcts in the following manner.
First, we may assume that each literal appears at most twice (and at least once), since a variable which appears only positively or only negatively can be preprocessed without consequence.
If a literal $x$ appears only once, then we add the following clauses using the new variables $y_1, \dots, y_5$ so that it appears twice.
\[
(x \lor y_1 \lor \overline{y}_2) \land
(y_1 \lor y_3 \lor \overline{y}_4) \land
(\overline{y}_1 \lor y_3 \lor \overline{y}_4) \land
(\overline{y}_1 \lor y_4 \lor \overline{y}_5) \land
(y_2 \lor y_4 \lor \overline{y}_5) \land
(y_2 \lor \overline{y}_3 \lor y_5) \land
(\overline{y}_2 \lor \overline{y}_3 \lor y_5)
\]

The additional variables $y_1, \dots, y_5$ all have exactly two appearances per literal, and the additional clauses can be satisfied by setting $y_1, \dots, y_5$ to true.
By repeatedly adding 5 new variables and 7 new clauses in this pattern for every literal which only appears once, we can ensure that we end up with an instance of \tcts.

For the remainder of this paper, we will also assume that no variable appears twice in a single clause in instances of \tcts and \rsat.
If a literal and its negation both appear, then the clause can be removed since it will be satisfied by any assignment.
If the same literal appears twice, then one appearance can be removed without changing the instance.
To adhere to the limitations of \rsat, we replace the removed variable with a new variable $y$ by duplicating the clause, adding $y$ to one copy, and adding $\overline{y}$ to the other.
We also assume that any clauses with only a single literal have been preprocessed.

\section{Computational Complexity}
In~\cite{grohe2018coloring}, the authors showed that \probname{Weak $r$-Orderable} is NP-hard for $r \geq 3$.
Below, we prove the stronger result that determining if $\wcol_r(G) \leq k$ is para-NP-hard parameterized by $k$, even when $r = 2$.
Moreover, we show that determining if $\col_r(G) \leq k$ is also para-NP-hard for $r \geq 2$.

\subsection{Weak 2-Orderable is Para-NP-hard}
In this section, we prove that \probname{Weak 2-Orderable} parameterized by the natural parameter $k$ is para-NP-hard.
This also closes an open question regarding the NP-hardness of \probname{Weak 2-Orderable} from~\cite{grohe2018coloring}.
Our proof reduces from \tcts and shows that determining if $\wcol_2(G) \leq 5$ is NP-hard.
First, we describe the reduction.

\begin{definition} \label{def:wcol2_redux}
	Given an instance $\varphi$ of \tcts, let $G(\varphi)$ be the following graph.
	For each clause $c_i \in \varphi$, we create 6 vertices $u_i^1, \dots, u_i^6$ in $G(\varphi)$.
	If $c_i$ only contains 2 literals, we add 2 more vertices $f_i$ and $f'_i$, each connected to $u_i^1, \dots, u_i^6$.
	Then for each variable $x_j$, we create 2 vertices $v_j$ and $v'_j$ (corresponding to the literals $x_j$ and $\overline{x}_j$ respectively) connected by an edge.
	For each variable $x_j$, add an edge from $v_j$ to each of $u_i^1, \dots, u_i^6$ for each clause $c_i$ which contains $x_j$ as a literal.
	Likewise, add an edge from $v'_j$ to each of $u_i^1, \dots, u_i^6$ for each clause $c_i$ which contains $\overline{x}_j$ as a literal.
\end{definition}

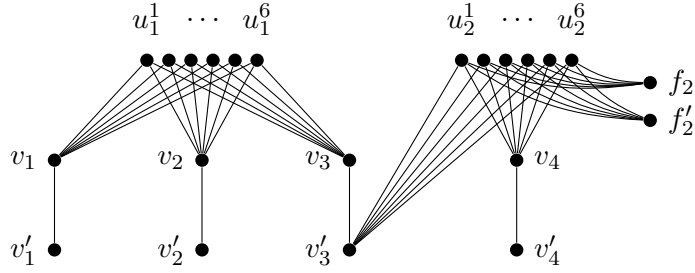
\begin{figure}
	\begin{center}
		\begin{tikzpicture}
	\tikzstyle{leafVert}=[circle,fill=black,minimum size=5pt,inner sep=0pt]

    \node[leafVert] (u11) {};
	\node[leafVert, right=.1cm of u11] (u12) {};
	\node[leafVert, right=.1cm of u12] (u13) {};
	\node[leafVert, right=.1cm of u13] (u14) {};
	\node[leafVert, right=.1cm of u14] (u15) {};
    \node[leafVert, right=.1cm of u15] (u16) {};
    \node[above=.1cm of u11] (l11) {$u_1^1$};
	\node[right=0.05cm of l11] (dots1) {$\dots$};
    \node[above=.1cm of u16] (l16) {$u_1^6$};
	\node[fit=(u11)(u16)(l11)(l16)] (a1) {};

    \node[leafVert, right=2.5cm of u16] (u21) {};
	\node[leafVert, right=.1cm of u21] (u22) {};
	\node[leafVert, right=.1cm of u22] (u23) {};
	\node[leafVert, right=.1cm of u23] (u24) {};
	\node[leafVert, right=.1cm of u24] (u25) {};
    \node[leafVert, right=.1cm of u25] (u26) {};
    \node[above=.1cm of u21] (l21) {$u_2^1$};
	\node[right=0.05cm of l21] (dots2) {$\dots$};
    \node[above=.1cm of u26] (l26) {$u_2^6$};
	\node[fit=(u21)(u26)(l21)(l26)] (a2) {};

	\node[leafVert, below=1cm of a1, label=left:{$v_2$}] (v2) {};
	\node[leafVert, below=1cm of v2, label=left:{$v'_2$}] (nv2) {};
	\draw (v2) -- (nv2);
	\draw (v2) -- (u11);
	\draw (v2) -- (u12);
	\draw (v2) -- (u13);
	\draw (v2) -- (u14);
	\draw (v2) -- (u15);
	\draw (v2) -- (u16);

    \node[leafVert, left=1.75cm of v2, label=left:{$v_1$}] (v1) {};
	\node[leafVert, below=1cm of v1, label=left:{$v'_1$}] (nv1) {};
	\draw (v1) -- (nv1);
	\draw (v1) -- (u11);
	\draw (v1) -- (u12);
	\draw (v1) -- (u13);
	\draw (v1) -- (u14);
	\draw (v1) -- (u15);
	\draw (v1) -- (u16);

    \node[leafVert, right=1.75cm of v2, label=left:{$v_3$}] (v3) {};
	\node[leafVert, below=1cm of v3, label=left:{$v'_3$}] (nv3) {};
	\draw (v3) -- (nv3);
	\draw (v3) -- (u11);
	\draw (v3) -- (u12);
	\draw (v3) -- (u13);
	\draw (v3) -- (u14);
	\draw (v3) -- (u15);
	\draw (v3) -- (u16);

	\draw (nv3) -- (u21);
	\draw (nv3) -- (u22);
	\draw (nv3) -- (u23);
	\draw (nv3) -- (u24);
	\draw (nv3) -- (u25);
	\draw (nv3) -- (u26);

    \node[leafVert, below=1cm of a2, label=right:{$v_4$}] (v4) {};
	\node[leafVert, below=1cm of v4, label=right:{$v'_4$}] (nv4) {};
	\draw (v4) -- (nv4);
	\draw (v4) -- (u21);
	\draw (v4) -- (u22);
	\draw (v4) -- (u23);
	\draw (v4) -- (u24);
	\draw (v4) -- (u25);
	\draw (v4) -- (u26);

	\node[leafVert, below right=0cm and 0.5cm of a2, label=right:{$f_2$}] (fu2) {};
	\node[leafVert, below right=0.5cm and 0.5cm of a2, label=right:{$f'_2$}] (fpu2) {};

	\draw (fu2) edge[bend left=15] (u21);
	\draw (fu2) edge[bend left=15] (u22);
	\draw (fu2) edge[bend left=15] (u23);
	\draw (fu2) edge[bend left=15] (u24);
	\draw (fu2) edge[bend left=15] (u25);
	\draw (fu2) edge[bend left=15] (u26);

	\draw (fpu2) edge[bend left=15] (u21);
	\draw (fpu2) edge[bend left=15] (u22);
	\draw (fpu2) edge[bend left=15] (u23);
	\draw (fpu2) edge[bend left=15] (u24);
	\draw (fpu2) edge[bend left=15] (u25);
	\draw (fpu2) edge[bend left=15] (u26);
\end{tikzpicture}
	\end{center}

	\caption{
		The subgraph of $G(\varphi)$ corresponding to the clauses $c_1 = (x_1 \lor x_2 \lor x_3)$ and $c_2 = (\overline{x}_3 \lor x_4)$.
	}
	\label{fig:wcol2_Gadget}
\end{figure}

See Figure~\ref{fig:wcol2_Gadget} for an example of $G(\varphi)$.
In the conversion between orders and assignments, the relative position of $v_j$ and $v'_j$ corresponds to the assignment of $x_j$.
If $v_j$ appears first, then $x_j$ is set to false, and if instead $v'_j$ appears first, then $x_j$ is set to true.
Thus, a clause vertex $u_i^\ell$ (which must appear at the beginning of the order) has 2-reach equal to 6 only if $c_i$ contains 3 false literals in the converted assignment.
We prove that $\varphi$ is satisfiable if and only if $\wcol_2(G(\varphi)) \leq 5$.
The forward direction uses a satisfying assignment of $\varphi$ to construct an order on $G(\varphi)$ with weak 2-reach at most 5.

\begin{lemma} \label{lem:sat-wcol2}
	Let $\varphi$ be an instance of \tcts and let $G(\varphi)$ be the graph produced by Definition~\ref{def:wcol2_redux}.
	If $\varphi$ is satisfiable, then $\wcol_2(G(\varphi)) \leq 5$.
\end{lemma}

\begin{proof}
	Let $G = G(\varphi)$, and let $A$ be an assignment of $x_1, \dots, x_n$ witnessing that $\varphi$ is satisfiable.
	Construct the total order $\sigma$ in the following manner.
	First, for each clause $c_i$, add $u_i^1, \dots, u_i^6$ to $\sigma$ in that order.
	If $c_i$ only contains 2 literals, place $f_i$ and $f'_i$ immediately after $u_i^6$.
	Then, for each variable $x_j$, add $v'_j$ and then $v_j$ if $x_j$ is true in $A$.
	Otherwise, add $v_j$ and then $v'_j$.

	Consider $\wreach_2(u_i^\ell, G_\sigma)$, and suppose that $c_i$ contains 3 literals.
	Since $u_i^\ell$ only has edges to literal vertices (of the form $v_j$ or $v'_j$) and $v_j, v'_j >_\sigma u_p^q$ for any $j$, it cannot reach any other $u_p^q$ (even if $p = i$).
	Thus, $|\wreach_2(u_i^\ell, G_\sigma)| \leq 6$ since there are only 6 vertices (corresponding to the 3 literals in the clause and their negations) remaining in its 2-neighborhood.
	This still holds if $c_i$ only contains 2 literals, since $f_i, f'_i >_\sigma u_i^6$ and neither $f_i$ nor $f'_i$ have edges outside of the clause.
	Since $A$ is a satisfying assignment, at least one literal in $c_i$ must be set to true.
	Without loss of generality, assume $x_j \in c_i$ and $x_j$ is true.
	Thus, $v_j$ appears after $v'_j$ by the definition of $\sigma$.
	Then $v'_j \not\in \wreach_2(u_i^\ell, G_\sigma)$, and so $|\wreach_2(u_i^\ell, G_\sigma)| \leq 5$.

	Next, consider $\wreach_2(f_i, G_\sigma)$.
	Since $u_i^\ell <_\sigma f_i$, $f_i$ can only reach $f'_i$ and vertices corresponding to literals in $c_i$.
	Thus, $|\wreach_2(f_i, G_\sigma)| = 3$ and similarly $|\wreach_2(f'_i, G_\sigma)| = 2$.

	Finally, consider $\wreach_2(v_j, G_\sigma)$.
	Since $u_i^\ell <_\sigma v_j$ and $f_i <_\sigma v_j$, $v_j$ can only reach other vertices corresponding to literals.
	If $x_j$ is false in $A$, then $v'_j$ is weakly 2-reachable from $v_j$.
	The only other vertices in the 2-neighborhood of $v_j$ are the other literals in the clauses which contain $x_j$.
	Since $\varphi$ is an instance of \tcts, there are only 4 of these.
	Thus, $|\wreach_2(v_j, G_\sigma)| \leq 5$.
	The argument is symmetric for $v'_j$.
	Since the weak 2-reach of every vertex is at most $5$, $\wcol_2(G_\sigma) \leq 5$.
\end{proof}

To complete the proof, the reverse direction argues that the assignment implied by an order $\sigma$ either satisfies $\varphi$, or there exists a subset of vertices (corresponding to an unsatisfied clause) which cannot have weak 2-reach at most 5 in $\sigma$.

\begin{lemma} \label{lem:wcol2-sat}
	Let $\varphi$ be an instance of \tcts and let $G(\varphi)$ be the graph produced by Definition~\ref{def:wcol2_redux}.
	If $\wcol_2(G(\varphi)) \leq 5$, then $\varphi$ has a satisfying assignment.
\end{lemma}

\begin{proof}
	Let $G = G(\varphi)$, and let $\sigma$ be a total order of $G$ witnessing that $\wcol_2(G_\sigma) \leq 5$.
	Let $A$ be the following assignment of $x_1, \dots, x_n$: if $v_j <_\sigma v'_j$, set $x_j$ to false in $A$; otherwise, set $x_j$ to true.
	Assume that $A$ does not satisfy $\varphi$ to produce a contradiction.

	Let $c_i$ be a clause in $\varphi$ which is not satisfied by $A$.
	First, suppose that $c_i$ contains 3 literals.
	Without loss of generality, let $c_i = (x_1 \lor x_2 \lor x_3)$.
	Since $c_i$ is not satisfied by $A$, $v_1 <_\sigma v'_1$, $v_2 <_\sigma v'_2$, and $v_3 <_\sigma v'_3$ by construction.
	Let $u_i^\ell$ be the vertex from $u_i^1, \dots, u_i^6$ which appears first in $\sigma$.
	If $u_i^\ell <_\sigma v_1, v_2, v_3$, then $|\wreach_2(u_i^\ell, G_\sigma)| \geq 6$ since it can reach $v_1$, $v_2$, $v_3$, $v'_1$, $v'_2$, and $v'_3$.
	This contradicts that $\wcol_2(G_\sigma) \leq 5$, and so we may assume that at least one of $v_1$, $v_2$, or $v_3$ appears before $u_i^\ell$ in $\sigma$.
	Without loss of generality, assume $v_1 <_\sigma u_i^\ell$.
	By our choice of $u_i^\ell$ though, $v_1$ can reach all of $u_i^1, \dots, u_i^6$.
	Thus, $|\wreach_2(v_1, G_\sigma)| \geq 6$.
	This again contradicts that $\wcol_2(G_\sigma) \leq 5$.

	If instead $c_i$ only contains 2 literals, then $u_i^\ell$ reaches $f_i$ and $f'_i$ rather than $v_3$ and $v'_3$.
	Like $v_1$ and $v_2$, $f_i$ and $f'_i$ cannot come before $u_i^\ell$ without reaching all of $u_i^1, \dots, u_i^6$, contradicting that $\wcol_2(G_\sigma) \leq 5$.
	Therefore, $A$ must be a satisfying assignment for $\varphi$.
\end{proof}

Together, Lemmas~\ref{lem:sat-wcol2} and~\ref{lem:wcol2-sat} show that determining if $\wcol_2(G) \leq 5$ is NP-hard.
Since 5 is a constant, this proves the main result.

\begin{theorem} \label{thm:wcol2-hardness}
	\probname{Weak 2-Orderable} parameterized by the natural parameter $k$ is para-NP-hard.
\end{theorem}

\begin{proof}
	Lemmas~\ref{lem:sat-wcol2} and~\ref{lem:wcol2-sat} imply that determining if $\wcol_2(G) \leq k$ is NP-hard even when $k = 5$.
\end{proof}

\begin{corollary} \label{cor:wcol2-apx-hardness}
	It is NP-hard to approximate \probname{Minimum Weak 2-Orderability} within a factor of $\frac{6}{5}$.
\end{corollary}

\begin{corollary} \label{cor:wcol2-xp-hardness}
	There does not exist an algorithm which can decide a given instance $(G, k)$ of \probname{Weak 2-Orderable} in $O(n^{f(k)})$ time unless P $=$ NP.
	Thus, \probname{Weak 2-Orderable} is not in XP.
\end{corollary}

\subsection{Weak $r$-Orderable is Para-NP-hard}
Having established the para-NP-hardness of the weak 2-coloring number, we now turn our attention to the weak $r$-coloring number for $r \geq 3$.
The construction is similar to the weak 2-coloring number, but includes subdivided edges to increase the distance between vertices.
An $\ell$-subdivided edge (alternatively: an edge subdivided $\ell$ times) is an induced path with $\ell$ internal vertices connecting the endpoints of the edge.
The vertices along the subdivided edge do not have additional neighbors outside of the path.
Note that a 0-subdivided edge is simply a normal edge.
We refer collectively to the vertices along the subdivided edges as \emph{subdivision vertices}.
Additional modifications handle the non-constant reach of vertices along subdivided edges.
One key adaptation is that we now reduce from \rsat.

\begin{definition} \label{def:wcol_redux}
	Given an instance $\varphi$ of \rsat, let $G(\varphi)$ be the following graph.
	For each clause $c_i \in \varphi$, we create $2r$ vertices $u_i^1, \dots, u_i^{2r}$ in $G(\varphi)$.
	Then for each variable $x_j$, we create 2 vertices $v_j$ and $v'_j$ (corresponding to the literals $x_j$ and $\overline{x}_j$ respectively) connected by an edge.
	For each variable $x_j$, add two $(r-2)$-subdivided edges from $v_j$ to each of $u_i^1, \dots, u_i^{2r}$ for each clause $c_i$ which contains $x_j$ as a literal.
	Likewise, add two $(r-2)$-subdivided edges from $v'_j$ to each of $u_i^1, \dots, u_i^{2r}$ for each clause $c_i$ which contains $\overline{x}_j$ as a literal.
\end{definition}

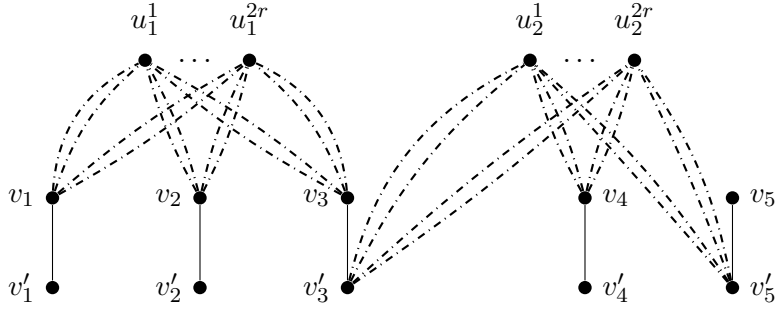
\begin{figure}
	\begin{center}
		\begin{tikzpicture}
	\tikzstyle{leafVert}=[circle,fill=black,minimum size=5pt,inner sep=0pt]

    \node[leafVert] (u11) {};
    \node[right=.2cm of u11] (dots1) {$\dots$};
    \node[leafVert, right=.2cm of dots1] (u16) {};
    \node[above=.1cm of u11] (l11) {$u_1^1$};
    \node[above=.1cm of u16] (l16) {$u_1^{2r}$};
	\node[fit=(u11)(u16)(l11)(l16)] (a1) {};

    \node[leafVert, right=3.5cm of u16] (u21) {};
    \node[right=.2cm of u21] (dots2) {$\dots$};
    \node[leafVert, right=.2cm of dots2] (u26) {};
    \node[above=.1cm of u21] (l21) {$u_2^1$};
    \node[above=.1cm of u26] (l26) {$u_2^{2r}$};
	\node[fit=(u21)(u26)(l21)(l26)] (a2) {};

	\node[leafVert, below=1.5cm of a1, label=left:{$v_2$}] (v2) {};
	\node[leafVert, below=1cm of v2, label=left:{$v'_2$}] (nv2) {};
	\draw (v2) -- (nv2);

    \node[leafVert, left=1.75cm of v2, label=left:{$v_1$}] (v1) {};
	\node[leafVert, below=1cm of v1, label=left:{$v'_1$}] (nv1) {};
	\draw (v1) -- (nv1);

    \node[leafVert, right=1.75cm of v2, label=left:{$v_3$}] (v3) {};
	\node[leafVert, below=1cm of v3, label=left:{$v'_3$}] (nv3) {};
	\draw (v3) -- (nv3);

    \node[leafVert, below=1.5cm of a2, label=right:{$v_4$}] (v4) {};
	\node[leafVert, below=1cm of v4, label=right:{$v'_4$}] (nv4) {};
	\draw (v4) -- (nv4);

	\node[leafVert, right=1.75cm of v4, label=right:{$v_5$}] (v5) {};
	\node[leafVert, below=1cm of v5, label=right:{$v'_5$}] (nv5) {};
	\draw (v5) -- (nv5);

	\draw[dash dot, thick] (v1) edge[bend left=30] (u11);
	\draw[dash dot, thick] (v1) edge[bend left=15] (u11);
	\draw[dash dot, thick] (v1) edge[bend left=6] (u16);
	\draw[dash dot, thick] (v1) edge[bend right=6] (u16);
	\draw[dash dot, thick] (v2) edge[bend left=7] (u11);
	\draw[dash dot, thick] (v2) edge[bend right=7] (u11);
	\draw[dash dot, thick] (v2) edge[bend left=7] (u16);
	\draw[dash dot, thick] (v2) edge[bend right=7] (u16);
	\draw[dash dot, thick] (v3) edge[bend left=6] (u11);
	\draw[dash dot, thick] (v3) edge[bend right=6] (u11);
	\draw[dash dot, thick] (v3) edge[bend right=30] (u16);
	\draw[dash dot, thick] (v3) edge[bend right=15] (u16);

	\draw[dash dot, thick] (nv3) edge[bend left=25] (u21);
	\draw[dash dot, thick] (nv3) edge[bend left=12] (u21);
	\draw[dash dot, thick] (nv3) edge[bend left=4] (u26);
	\draw[dash dot, thick] (nv3) edge[bend right=4] (u26);

	\draw[dash dot, thick] (v4) edge[bend left=7] (u21);
	\draw[dash dot, thick] (v4) edge[bend right=7] (u21);
	\draw[dash dot, thick] (v4) edge[bend left=7] (u26);
	\draw[dash dot, thick] (v4) edge[bend right=7] (u26);

	\draw[dash dot, thick] (nv5) edge[bend right=8] (u21);
	\draw[dash dot, thick] (nv5) edge[bend right=1] (u21);
	\draw[dash dot, thick] (nv5) edge[bend right=12] (u26);
	\draw[dash dot, thick] (nv5) edge[bend right=3] (u26);
\end{tikzpicture}
	\end{center}

	\caption{
		The subgraph of $G(\varphi)$ corresponding to the \probname{3-SAT} clauses $c_1 = (x_1 \lor x_2 \lor x_3)$ and $c_2 = (\overline{x}_3 \lor x_4 \lor \overline{x}_5)$.
		The dashed edges represent $(r-2)$-subdivided edges between $u_i^\ell$ and $v_j$ (or $v'_j$).
	}
	\label{fig:wcol_Gadget}
\end{figure}

See Figure~\ref{fig:wcol_Gadget} for an example of $G(\varphi)$.
Like the weak 2-coloring number reduction, the conversion between orders and assignments conflates the relative position of $v_j$ and $v'_j$ with the assignment of $x_j$.
If $v_j$ appears first, then $x_j$ is set to false, and if instead $v'_j$ appears first, then $x_j$ is set to true.
Thus, a clause vertex $u_i^\ell$ (which still must appear at the beginning of the order) has $r$-reach equal to $2r$ only if $c_i$ contains $r$ false literals in the converted assignment.
We prove that $\varphi$ is satisfiable if and only if $\wcol_r(G(\varphi)) \leq 2r-1$.
The forward direction uses a satisfying assignment of $\varphi$ to construct an order on $G(\varphi)$ with weak $r$-reach at most $2r-1$.

\begin{lemma} \label{lem:sat-wcol}
	Let $\varphi$ be an instance of \rsat and let $G(\varphi)$ be the graph produced by Definition~\ref{def:wcol_redux}.
	If $\varphi$ is satisfiable, then $\wcol_r(G(\varphi)) \leq 2r-1$.
\end{lemma}

\begin{proof}
	Let $G = G(\varphi)$, and let $A$ be an assignment of $x_1, \dots, x_n$ witnessing that $\varphi$ is satisfiable.
	Construct the total order $\sigma$ in the following manner.
	First, place the subdivision vertices in any order.
	For each clause $c_i$, add $u_i^1, \dots, u_i^{2r}$ to $\sigma$ in that order.
	Then, for each variable $x_j$, add $v'_j$ and then $v_j$ if $x_j$ is true in $A$.
	Otherwise, add $v_j$ and then $v'_j$.

	Consider a subdivision vertex $s$ along a subdivided edge between $u_i^\ell$ and $v_j$.
	Depending on the order, $s$ may be able to reach the other $r-3$ vertices along its subdivided edge.
	However, it cannot reach a subdivision vertex on a different subdivided edge since that would require going through $u_i^\ell$ or $v_j$, both of which appear after all subdivision vertices.
	The $r$-reach of $s$ contains only $u_i^\ell$ and no other clause vertices since it cannot use $v_j$ which appears after all clause vertices in $\sigma$.
	Finally, $s$ can reach $v_j$ and $v'_j$.
	If $s$ is adjacent to $u_i^\ell$, it can reach the $r-1$ vertices corresponding to the other literals in $c_i$.
	In total then, $|\wreach_r(s, G_\sigma)| \leq 2r-1$.

	Next, consider $\wreach_r(u_i^\ell, G_\sigma)$.
	Since edges between clause vertices and literal vertices are subdivided $r-2$ times and $r \geq 3$, the distance between $u_i^\ell$ and $u_p^q$ is greater than $r$ (even if $p = i$).
	Thus, $|\wreach_r(u_i^\ell, G_\sigma)| \leq 2r$ since there are only $2r$ vertices (corresponding to the $r$ literals in the clause and their negations) remaining in its $r$-neighborhood.
	Since $A$ is a satisfying assignment, at least one literal in $c_i$ must be set to true.
	Without loss of generality, assume $x_j \in c_i$ and $x_j$ is true.
	Thus, $v_j$ appears after $v'_j$ by the definition of $\sigma$.
	Then $v'_j \not\in \wreach_r(u_i^\ell, G_\sigma)$, and so $|\wreach_r(u_i^\ell, G_\sigma)| \leq 2r-1$.

	Finally, $|\wreach_r(v_j, G_\sigma)| \leq 1$.
	Although $v_j$ can reach $v'_j$ (or vice versa), every other vertex corresponding to a literal is at least distance $r+1$ away.
	Since the weak $r$-reach of every vertex is at most $2r-1$, $\wcol_r(G_\sigma) \leq 2r-1$.
\end{proof}

To complete the proof, the reverse direction argues that the assignment implied by an order $\sigma$ either satisfies $\varphi$, or there exists a subset of vertices (corresponding to an unsatisfied clause) which cannot have weak $r$-reach at most $2r-1$ in $\sigma$.

\begin{lemma} \label{lem:wcol-sat}
	Let $\varphi$ be an instance of \rsat and let $G(\varphi)$ be the graph produced by Definition~\ref{def:wcol_redux}.
	If $\wcol_r(G(\varphi)) \leq 2r-1$, then $\varphi$ has a satisfying assignment.
\end{lemma}

\begin{proof}
	Let $G = G(\varphi)$, and let $\sigma$ be a total order of $G$ witnessing that $\wcol_r(G_\sigma) \leq 2r-1$.
	Let $A$ be the following assignment of $x_1, \dots, x_n$: if $v_j <_\sigma v'_j$, set $x_j$ to false in $A$; otherwise, set $x_j$ to true.
	Assume that $A$ does not satisfy $\varphi$ to produce a contradiction.

	Let $c_i$ be a clause in $\varphi$ which is not satisfied by $A$.
	Without loss of generality, let $c_i = (x_1 \lor \dots \lor x_r)$.
	Since $c_i$ is not satisfied by $A$, $v_j <_\sigma v'_j$ for $1 \leq j \leq r$.
	Let $u_i^\ell$ be the vertex from $u_i^1, \dots, u_i^{2r}$ which appears first in $\sigma$.
	If $u_i^\ell <_\sigma v_1, \dots, v_r$, then $|\wreach_r(u_i^\ell, G_\sigma)| \geq 2r$ since it can reach $v_1, \dots v_r$ and $v'_1, \dots, v'_r$.
	Reaching $v_j$ and $v'_j$ can only be avoided by placing a subdivision vertex $s$ from the subdivided edge between $u_i^\ell$ and $v_j$ after $v'_j$ in $\sigma$.
	However, since there are 2 subdivided edges between $u_i^\ell$ and $v_j$, this results in the same total reach.
	This contradicts that $\wcol_r(G_\sigma) \leq 2r-1$, and so we may assume that at least one of $v_1, \dots, v_r$ appears before $u_i^\ell$ in $\sigma$.
	Without loss of generality, assume $v_1 <_\sigma u_i^\ell$.
	By our choice of $u_i^\ell$, $v_1$ can reach all of $u_i^1, \dots, u_i^{2r}$.
	Again, this can be avoided by placing 2 subdivision vertices after $u_i^\ell$ in $\sigma$, but this yields an even larger reach.
	Thus, $|\wreach_r(v_1, G_\sigma)| \geq 2r$.
	This again contradicts that $\wcol_r(G_\sigma) \leq 2r-1$.
	Therefore, $A$ must be a satisfying assignment for $\varphi$.
\end{proof}

Together, Lemmas~\ref{lem:sat-wcol} and~\ref{lem:wcol-sat} show that determining $\wcol_r(G) \leq 2r-1$ is NP-hard.
Since $2r-1$ is assumed to be a constant, this proves the main result.

\begin{theorem} \label{thm:wcol-hardness}
	\probname{Weak $r$-Orderable} parameterized by the natural parameter $k$ is para-NP-hard.
\end{theorem}

\begin{proof}
	Lemmas~\ref{lem:sat-wcol} and~\ref{lem:wcol-sat} imply that determining if $\wcol_r(G) \leq k$ is NP-hard even when $k = 2r-1$.
\end{proof}

\begin{corollary} \label{cor:wcol-apx-hardness}
	It is NP-hard to approximate \probname{Minimum Weak $r$-Orderability} within a factor of $\frac{2r}{2r-1}$.
\end{corollary}

\begin{corollary} \label{cor:wcol-xp-hardness}
	For any fixed $r$, there does not exist an algorithm which can decide a given instance $(G, k)$ of \probname{Weak $r$-Orderable} in $O(n^{f(k)})$ time unless P = NP.
	Thus, \probname{Weak $r$-Orderable} is not in XP.
\end{corollary}

\subsection{$r$-Orderable is Para-NP-Hard}
Finally, we prove that \probname{$r$-Orderable} parameterized by the natural parameter $k$ is para-NP-hard.
Our proof reduces from \tcts{} and shows that determining if $\col_r(G) \leq 6$ is NP-hard.
First, we describe the reduction from the \tcts{} instance $\varphi$ to the graph $G(\varphi)$.

\begin{definition} \label{def:col_redux}
    Given an instance $\varphi$ of \tcts{}, let $G(\varphi)$ be the following graph.
    For each clause $c_i \in \varphi$, create a \emph{clause vertex} $u_i$ in $G(\varphi)$.
    For each variable $x_j$, create 2 \emph{literal vertices} $v_j$ and $v'_j$ (corresponding to the literals $x_j$ and $\overline{x}_j$ respectively) and connect them by an edge.
    Additionally, for each clause $c_i$ which contains the literal $x_j$ (or $\overline{x}_j$), connect $u_i$ and $v_j$ (or $v'_j$) with an $(r-1)$-subdivided edge.
	To complete $G(\varphi)$, create a 7-clique consisting of the \emph{clique vertices} $w_1, \dots, w_7$, and connect it to the remainder of the graph in the following manner.
	For every clause $c_i$, add edges from $u_i$ to $w_1, \dots, w_4$.
	If $c_i$ only contains 2 literals, add an additional edge from $u_i$ to $w_5$.
	For every variable $x_j$, add edges from $v_j$ to $w_2, \dots, w_4$ and from $v'_j$ to $w_5, \dots, w_7$.
\end{definition}

Figure~\ref{fig:col_gadget} gives an example of $G(\varphi)$ for a small \tcts{} instance.
In the conversion between orders and assignments, the relative position of $v_j$ and $v'_j$ corresponds to the assignment of $x_j$.
If $v_j$ appears first, then $x_j$ is set to true, and if instead $v'_j$ appears first, then $x_j$ is set to false.
Informally, the reduction works by forcing the clause vertices to appear between the true and false literal vertices in optimal orders.
As a result, a clause vertex has reach 7 unless it is satisfied (3 false literal vertices plus 4 clique vertices).
Formally, we prove that $\varphi$ is satisfiable if and only if $\col_r(G(\varphi)) \leq 6$.
The forward direction uses a satisfying assignment of $\varphi$ to construct an order on $G(\varphi)$ with $r$-reach at most 6.

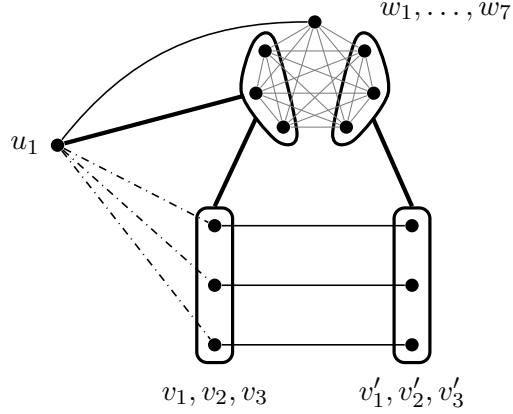
\begin{figure}[t]
	\begin{center}
		\begin{tikzpicture}
	\tikzstyle{leafVert}=[circle,fill=black,minimum size=5pt,inner sep=0pt]

    \node[leafVert, label=left:{$u_1$}] (u1) {};

	\node[leafVert, below right=1cm and 2cm of u1.center] (v1) {};
	\node[leafVert, below=0.6cm of v1] (v2) {};
	\node[leafVert, below=0.6cm of v2] (v3) {};
	\draw[dash dot, semithick] (u1) -- (v1);
	\draw[dash dot, semithick] (u1) -- (v2);
	\draw[dash dot, semithick] (u1) -- (v3);
	\node[fit=(v1)(v2)(v3), draw, very thick, rounded corners, inner sep=4pt] (v_pos) {};

	\node[leafVert, right=2.5cm of v1.center] (vp1) {};
	\node[leafVert, right=2.5cm of v2.center] (vp2) {};
	\node[leafVert, right=2.5cm of v3.center] (vp3) {};
	\draw[semithick] (v1) -- (vp1);
	\draw[semithick] (v2) -- (vp2);
	\draw[semithick] (v3) -- (vp3);
	\node[fit=(vp1)(vp2)(vp3), draw, very thick, rounded corners, inner sep=4pt] (v_neg) {};

	\node[leafVert, above right=1.5cm and 3.25cm of u1] (w1) {};
	\node[leafVert, below left=0.25cm and 0.52cm of w1] (w2) {};
	\node[leafVert, below left=0.81cm and 0.64cm of w1] (w3) {};
	\node[leafVert, below left=1.26cm and 0.28cm of w1] (w4) {};
	\node[leafVert, below right=0.25cm and 0.52cm of w1] (w5) {};
	\node[leafVert, below right=0.81cm and 0.64cm of w1] (w6) {};
	\node[leafVert, below right=1.26cm and 0.28cm of w1] (w7) {};

	\node[fit=(w2)(w3)(w4), inner sep=1.5pt] (w_pos1) {};
	\node[fit=(w2)(w3)(w4), inner sep=-3pt] (w_pos2) {};
	\node[fit=(w5)(w6)(w7), inner sep=-3pt] (w_neg) {};


	\draw[very thick] plot [smooth cycle] coordinates {
		($(w2.center) + (0.12, 0.15)$)
		($(w2.center) + (-0.15, 0.12)$)
		($(w3.center) + (-0.19, 0)$)
		($(w4.center) + (-0.09, -0.18)$)
		($(w4.center) + (0.18, -0.09)$)
	};

	\draw[very thick] plot [smooth cycle] coordinates {
		($(w5.center) + (-0.12, 0.15)$)
		($(w5.center) + (0.15, 0.12)$)
		($(w6.center) + (0.19, 0)$)
		($(w7.center) + (0.09, -0.18)$)
		($(w7.center) + (-0.18, -0.09)$)
	};

	\draw (u1) edge[bend left=26, semithick] (w1);
	\draw[ultra thick] (u1) -- (w_pos1);
	\draw[ultra thick] (v_pos.north) -- (w_pos2);
	\draw[ultra thick] (v_neg.north) -- (w_neg);

	\draw[gray] (w1) -- (w2);
	\draw[gray] (w1) -- (w3);
	\draw[gray] (w1) -- (w4);
	\draw[gray] (w1) -- (w5);
	\draw[gray] (w1) -- (w6);
	\draw[gray] (w1) -- (w7);
	\draw[gray] (w2) -- (w3);
	\draw[gray] (w2) -- (w4);
	\draw[gray] (w2) -- (w5);
	\draw[gray] (w2) -- (w6);
	\draw[gray] (w2) -- (w7);
	\draw[gray] (w3) -- (w4);
	\draw[gray] (w3) -- (w5);
	\draw[gray] (w3) -- (w6);
	\draw[gray] (w3) -- (w7);
	\draw[gray] (w4) -- (w5);
	\draw[gray] (w4) -- (w6);
	\draw[gray] (w4) -- (w7);
	\draw[gray] (w5) -- (w6);
	\draw[gray] (w5) -- (w7);
	\draw[gray] (w6) -- (w7);

	\node[above right=0.2cm and 0cm of w5] {$w_1, \dots, w_7$};
	\node[below=0.2cm of vp3] (vplab) {$v'_1, v'_2, v'_3$};
	\node[left=2.59cm of vplab.text, anchor=text] {$v_1, v_2, v_3$};
\end{tikzpicture}
	\end{center}

	\caption{
		The subgraph of $G(\varphi)$ corresponding to the clause $c_1 = (x_1 \lor x_2 \lor x_3)$.
		The dashed edges represent $(r-1)$-subdivided edges between clause vertices and literal vertices.
		The bold edges indicate that every possible edge is present between the two sets.
	}
	\label{fig:col_gadget}
\end{figure}

\begin{lemma} \label{lem:sat-col}
	Let $\varphi$ be an instance of \tcts{}, and let $G(\varphi)$ be the graph defined by Definition~\ref{def:col_redux}.
	If $\varphi$ is satisfiable, then $\col_r(G(\varphi)) \leq 6$.
\end{lemma}

\begin{proof}
	Let $G = G(\varphi)$, and let $A$ be an assignment of $x_1, \dots, x_n$ satisfying $\varphi$.
	Construct the total order $\sigma$ in the following manner.
	First, place the subdivision vertices at the beginning of $\sigma$.
	Then for each literal $x_j$ (or $\overline{x}_j$) which is true in $A$, add $v_j$ (or $v'_j$) to $\sigma$.
	Next, add $u_i$ to $\sigma$ for each clause $c_i$.
	Finally, add the remaining vertices corresponding to false literals in $A$ and then $w_1, \dots, w_7$ to the end of $\sigma$.

	First, we note that the subdivision vertices all have reach 2 since they only reach a single vertex in either direction along the subdivided edge.
	Consider $\reach_r(v_j, G_\sigma)$ where $v_j$ corresponds to a true literal $x_j$ in $A$.
	In this case, $v_j$ reaches $v'_j$, both vertices corresponding to the clauses containing $x_j$, and 3 clique vertices for a total of 6.
	Since every path leaving $v_j$ goes through one of these vertices before any others (not counting subdivision vertices), $v_j$ cannot reach anything else.
	Thus, $|\reach_r(v_j, G_\sigma)| \leq 6$ when $x_j$ is true in $A$.
	The argument is symmetric for $v'_j$ when $\overline{x}_j$ is true in $A$.

	Next, consider $\reach_r(u_i, G_\sigma)$.
	Since $u_i$ is connected to literal vertices using $(r-1)$-subdivided edges, $u_i$ does not have any qualifying paths which go through a true literal vertex.
	As a result, $u_i$ reaches $w_1, \dots, w_4$ and at most 2 false literal vertices since $A$ satisfies $\varphi$.
	If $c_i$ only contains 2 literals, then $u_i$ also reaches $w_5$, but it can only reach a single false literal vertex.
	Thus, $|\reach_r(u_i, G_\sigma)| \leq 6$.

	Finally, consider $\reach_r(v_j, G_\sigma)$, but now suppose $x_j$ is false in $A$.
	Every path from $v_j$ to another false literal vertex which does not go through a clique vertex has length greater than $r$. 
	Thus, $v_j$ can only reach $w_2, \dots, w_7$ (either directly or through $v'_j$), and $|\reach_r(v_j, G_\sigma)| \leq 6$ when $x_j$ is false in $A$.
	The argument is symmetric for $v'_j$ when $\overline{x}_j$ is false in $A$.
	The 7 clique vertices trivially have reach at most 6 since no other vertices appear after them in $\sigma$.
	Because every vertex has bounded reach with respect to $\sigma$, $\col_r(G_\sigma) \leq 6$.
\end{proof}

To complete the proof, the reverse direction argues that the assignment implied by an order $\sigma$ either satisfies $\varphi$, or there exists a subset of vertices (corresponding to an unsatisfied clause) which cannot be ordered with $r$-reach at most 6.

\begin{lemma} \label{lem:col-sat}
	Let $\varphi$ be an instance of \tcts{}, and let $G(\varphi)$ be the graph defined by Definition~\ref{def:col_redux}.
	If $\col_r(G(\varphi)) \leq 6$, then $\varphi$ has a satisfying assignment.
\end{lemma}

\begin{proof}
	Let $G = G(\varphi)$, and let $\sigma$ be a total order of $G$ such that $\col_r(G_\sigma) \leq 6$.
	Let $A$ be the following assignment of $x_1, \dots, x_n$: if $v_j <_\sigma v'_j$, set $x_j$ to true in $A$; otherwise, set $x_j$ to false.
	Assume that $A$ does not satisfy $\varphi$ to produce a contradiction.

	Let $c_i$ be a clause in $\varphi$ which is not satisfied by $A$.
	Without loss of generality, let $v_j$ be the first vertex in $\sigma$ such that $x_j \in c_i$.
	Since $c_i$ is not satisfied by $A$, $v'_j <_\sigma v_j$.
	Let $w_\ell$ be the vertex from $w_1, \dots, w_7$ which appears first in $\sigma$.
	We argue that the first vertex among $u_i$, $v_j$, and $w_\ell$ to appear in $\sigma$ must have $r$-reach at least 7.

	First, suppose $w_\ell <_\sigma u_i, v_j$.
	By our choice of $w_\ell$, $w_\ell$ reaches the other 6 clique vertices.
	If $\ell \leq 4$, then $w_\ell$ also reaches $u_i$, so we may assume $\ell > 4$.
	However, in this case, $w_\ell$ reaches $v'_j$ if $w_\ell <_\sigma v'_j$, or it reaches $v_j$ if $w_\ell >_\sigma v'_j$.
	Thus, $|\reach_r(w_\ell, G_\sigma)| > 6$.
	Suppose instead that $u_i <_\sigma v_j, w_\ell$.
	By our choice of $w_\ell$, $u_i$ must reach $w_1, \dots, w_4$.
	Moreover, by our choice of $v_j$, $u_i$ must reach each vertex corresponding to one of the 3 literals in $c_i$ (or a vertex along the subdivided edge) for a total of 7.
	This is not improved if $c_i$ only contains 2 literals since then $u_i$ also reaches $w_5$.
	Thus, $|\reach_r(u_i, G_\sigma)| > 6$.
	Finally, suppose that $v_j <_\sigma u_i, w_\ell$.
	In this case, $v_j$ reaches $u_i$ as well as $w_2, \dots, w_7$ since $v'_j <_\sigma v_j$.
	Thus, $|\reach_r(v_j, G_\sigma)| > 6$.
	Since one of these vertices must appear before the others, $\col_r(G_\sigma)$ cannot be at most 6.
\end{proof}

Together, Lemmas~\ref{lem:sat-col} and~\ref{lem:col-sat} show that determining $\col_r(G) \leq 6$ is NP-hard.
Since 6 is a constant, this proves the main result.

\begin{theorem} \label{thm:col_paranphard}
	When parameterized by the natural parameter $k$, \probname{$r$-Orderable} is para-NP-hard.
\end{theorem}

\begin{proof}
    Lemmas~\ref{lem:sat-col} and~\ref{lem:col-sat} imply that determining if $\col_r(G) \leq k$ is NP-hard even when $k = 6$.
\end{proof}

\begin{corollary} \label{cor:col-apx-hardness}
	It is NP-hard to approximate \probname{Minimum $r$-Orderability} within a factor of $\frac{7}{6}$.
\end{corollary}

\begin{corollary} \label{cor:col-xp-hardness}
	There does not exist an algorithm which can decide a given instance $(G, k)$ of \probname{$r$-Orderable} in $O(n^{f(k)})$ time unless P = NP.
	Thus, \probname{$r$-Orderable} is not in XP.
\end{corollary}

\section{Approximating the Generalized Coloring Numbers}

In this section, we prove that a straightforward greedy algorithm for \probname{Minimum $r$-Orderability} gives a $(k-1)^{r-1}$-approximation in graphs with $r$-admissibility at most $k$.
We establish that the same algorithm also gives an $O(k^{r-1})$-approximation for \probname{Minimum Weak $r$-Orderability}.
Our analysis relies on properties of $r$-admissibility and its relationships with the $r$-coloring number and the weak $r$-coloring number.

\begin{definition} \label{def:adm}
    Given a graph $G = (V, E)$ and a prefix order $\sigma$, the \textbf{$r$-backconnectivity} $\bcon_r(u, G_\sigma)$ of a vertex $u \in V$ is the maximum number of vertex-disjoint (aside from $u$) $r$-qualifying paths beginning at $u$.
    The \textbf{$r$-admissibility} of $G$ is defined as
    \[ \adm_r(G) := \min_{\sigma \in \Pi} \max_{u \in V}\ \bcon_r(u, G_\sigma). \]
\end{definition}

We note again that this definition adheres to a left-to-right convention, reversing the direction used in prior work.
Unlike $\reach_r(u, G_\sigma)$ and $\wreach_r(u, G_\sigma)$ which are sets, notice that $\bcon_r(u, G_\sigma)$ is simply a non-negative integer, and a set of disjoint paths witnessing $\bcon_r(u, G_\sigma)$ need not be unique.
By definition, $\adm_r(G_\sigma) \leq \col_r(G_\sigma) \leq \wcol_r(G_\sigma)$.
In~\cite{dvorak2013constant}, Dvo\v{r}\'{a}k observed that bounds also hold in the opposite direction.

\begin{lemma}{\cite{dvorak2013constant}} \label{lem:adm_col}
	If $\adm_r(G_\sigma) \leq k$, then $\col_r(G_\sigma) \leq k \cdot (k-1)^{r-1}$ and $\wcol_r(G_\sigma) \leq \frac{k^{r+1} - 1}{k-1}$.
\end{lemma}

Computing the $r$-backconnectivity of a vertex is NP-hard~\cite{itai1982complexity}, and so Dvo\v{r}\'{a}k's algorithm for the generalized coloring numbers approximates the $r$-admissibility before applying these relationships.
This approach yields an $(r \cdot (rk-1)^{r-1})$-approximation for the $r$-coloring number and an $O((rk)^{r-1})$-approximation for the weak $r$-coloring number.
We show that the approximation factors can be improved to $(k-1)^{r-1}$ and $O(k^{r-1})$ respectively by applying a similar analysis directly to a greedy algorithm.
The greedy algorithm we use generalizes the strategy used to compute degeneracy and constructs an order $\sigma$ by repeatedly selecting the vertex with the smallest \emph{estimated} $r$-backconnectivity.

\begin{definition} \label{def:estimate}
    Given a graph $G = (V, E)$ and a prefix order $\sigma$, an $\ell$-qualifying $u$-$v$ path is \emph{shortest} if there does not exist an $(\ell-1)$-qualifying $u$-$v$ path.
    The \textbf{estimated} $r$-backconnectivity of $u$, denoted $\est_r(u, G_\sigma)$, is the maximum number of disjoint (aside from $u$) shortest $r$-qualifying paths.
\end{definition}

Clearly, $\est_r(u, G_\sigma) \leq \bcon_r(u, G_\sigma)$ since any paths $P_1, \dots, P_\ell$ witnessing that $\est_r(u, G_\sigma) \geq \ell$ also witness that $\bcon_r(u, G_\sigma) \geq \ell$.
The estimated $r$-backconnectivity of $u$ can be computed in $O(kn)$ time by computing a maximum flow in a directed acyclic graph derived from the level graph rooted at $u$.
Since every source-sink path in this graph is necessarily shortest and has length at most $r$ by construction, flow techniques work in this setting even though they fail for general $r$-backconnectivity (where the maximum flow can exceed $\bcon_r(u, G_\sigma)$ by using a path with length greater than $r$).

\IncMargin{1em}
\begin{algorithm}
    \SetKwInOut{Input}{input}
	\SetKwInOut{Output}{output}
	\Input{a graph $G$ with $\adm_r(G) = k$ and an integer $r>0$.}
    \Output{a total order $\sigma$ with $\col_r(G_\sigma) \leq k \cdot (k-1)^{r-1}$ and $\wcol_r(G_\sigma) \leq \frac{k^{r+1} - 1}{k-1}$.}
    \BlankLine
    $S\leftarrow \{u : \deg(u)\ |\ u \in V(G)\}$ \\
    $\sigma\leftarrow [\emptyset]$ \\
    \While{$S$ is not empty}{

		$u \leftarrow$ arg\,min $S$

		append $u$ to $\sigma$ and remove $u$ from $S$

		\For{$v$ in $\reach_r(u, G_\sigma)$}{
			$S[v] = \est_r(v, G_\sigma)$
		}
    }
    \Return{$\sigma$}
    \caption{BoundedColoring($G,r$)}
    \label{alg:main}
\end{algorithm}

Algorithm~\ref{alg:main} describes our approach in detail.
At each step, the algorithm selects the vertex with minimum estimated $r$-backconnectivity in constant time using a bucket queue.
Then, the algorithm updates the estimated $r$-backconnectivity for each vertex that may have been affected by ordering $u$.
Since these are exactly the vertices in $\reach_r(u, G_\sigma)$, Algorithm~\ref{alg:main} takes $O(tkn^2)$ time in total when $\col_r(G_\sigma) = t$.
To prove the correctness of Algorithm~\ref{alg:main}, we first show that for any prefix order, there exists an unordered vertex with small backconnectivity.

\begin{lemma} \label{lem:bc_noninc}
	Given a graph $G = (V, E)$ and a prefix order $\sigma$, there exists an unordered vertex $u$ with respect to $\sigma$ such that $\bcon_r(u, G_\sigma) \leq \adm_r(G) = k$.
\end{lemma}

\begin{proof}
	Let $\xi = (v_1, \dots, v_n)$ be a total order of $G$ such that $\adm_r(G_\xi) = k$, and let $v_i$ be the first vertex in $\xi$ which is unordered in $\sigma$.
	By our choice of $v_i$, $\{v_1, \dots, v_{i-1}\}$ are ordered in $\sigma$.
	Let $P_1, \dots, P_\ell$ be a set of disjoint $r$-qualifying paths which witness that $\bcon_r(v_i, G_\sigma) = \ell$.
	Let $P'_j$ be the subpath of $P_j$ beginning at $v_i$ and ending at the first vertex along $P_j$ which is not in $\{v_1, \dots, v_{i-1}\}$.
	Since $P_j$ is an $r$-qualifying path from $v_i$, its endpoint must be in $\{v_{i+1}, \dots, v_n\}$, and so the path $P'_j$ is well-defined.
	The paths $P'_1, \dots, P'_\ell$ witness that $\bcon_r(v_i, G_\xi) \geq \ell$.
	Thus, $\bcon_r(v_i, G_\sigma) \leq \bcon_r(v_i, G_\xi) \leq \adm_r(G_\xi) = k$.
\end{proof}

Lemma~\ref{lem:bc_noninc} implies that the estimated backconnectivity of some unordered vertex will always be at most the admissibility.
Now, we prove that if the estimated $r$-backconnectivity is small, then the $r$-reach is also bounded.
To state the result precisely, we need additional notation.
Given a graph $G = (V, E)$, a prefix order $\sigma$, and a vertex $u$, let $d_v(u, G_\sigma)$ be the minimum $i$ such that $v$ is $i$-reachable from $u$.
We will abbreviate $d_v(u, G_\sigma)$ as $d_v$ when $u$ is clear from context.

\begin{lemma} \label{lem:col_bound}
	Let $G = (V, E)$ be a graph, and let $k = \adm_r(G)$.
	If $\sigma$ is a total order on $V$ such that $\est_r(u, G_\sigma) \leq k$ for all $u$, then
	\[ \sum_{v \in \reach_r(u, G_\sigma)} (k-1)^{r-d_v} \leq k \cdot (k-1)^{r-1}. \]
\end{lemma}

\begin{proof}
    Let $T \subseteq G$ be a tree rooted at $u$ such that the internal vertices of $T$ appear before $u$ in $\sigma$ and the leaves of $T$ at depth at most $i$ are exactly the vertices in $\reach_i(u, G_\sigma)$.
	Such a tree can be constructed by combining shortest $r$-qualifying $u$-$v$ paths for every $v \in \reach_r(u, G_\sigma)$.
    Suppose there exists a vertex $v \in T$ such that $\deg_T(v) > k$.
    By the definition of $T$, the paths from $v$ to $u$ and from $v$ to the leaves in the subtree rooted at $v$ must be shortest and have length at most $r$.
    At least one vertex along each of these paths must be reachable from $v$ since $u$ and all of the vertices in $\reach_r(u, G_\sigma)$ appear after $v$, since $v <_\sigma u <_\sigma w \in \reach_r(u, G_\sigma)$ by definition.
    Thus, $\est_r(v, G_\sigma) > k$, contradicting our assumption on $\sigma$, and so the vertices of $T$ have degree at most $k$.

    Let $L_v$ denote the set of leaves in the subtree of $T$ rooted at $v$.
    We claim that for every vertex $v \neq u$ in $T$, $\sum_{w \in L_v} (k-1)^{r-d_w} \leq (k-1)^{r-\depth_T(v)}$.
    Assume there exists a vertex $v \neq u$ in $T$ such that the claim does not hold.
    We may assume that $v$ is the vertex maximizing $\depth_T(v)$ out of all such vertices.
    The vertex $v$ cannot be a leaf since then $L_v$ contains only $v$.
    Since we chose $v$ to be the deepest vertex, the claim must hold for the children of $v$.
    Combined with the fact that $v$ has at most $k-1$ children and $\depth_T(v)$ is exactly one less than the depth of each child, the claim must also hold for $v$.

    We can apply the same logic to the root $u$ which has at most $k$ children in $T$, and so $\sum_{w \in L_u} (k-1)^{r-d_w} \leq k \cdot (k-1)^{r-1}$.
    The construction of $T$ implies that $L_u = \reach_r(u, G_\sigma)$, and so the lemma holds as desired.
\end{proof}

Since $(k-1)^{r-d_v} \geq 1$, this proves that $\col_r(G_\sigma) \leq k \cdot (k-1)^{r-1}$ as desired.
To complete the correctness proof of Algorithm~\ref{alg:main}, we only need to demonstrate the weak $r$-coloring number bound.
As noted in Dvo\v{r}\'{a}k's proof of Lemma 6 in~\cite{dvorak2013constant}, the weak $r$-coloring number bound follows directly from the bound on the $r$-coloring number.
We replicate the argument here in our notation for convenience.

\begin{observation}{\cite{dvorak2013constant}} \label{obs:wreach}
	Given a graph $G = (V, E)$ and a total order $\sigma$, the following holds for all $u \in V$:
	\[ \wreach_r(u, G_\sigma) = \bigcup_{v \in \reach_r(u, G_\sigma)} \{v\} \cup \wreach_{r-d_v}(v, G_\sigma). \]
\end{observation}

To see this, consider a weak $r$-qualifying path $P$ from $u$ to $w$.
Let $v$ be the first vertex on $P$ to appear after $u$ in $\sigma$.
The vertex $v$ divides $P$ into $P_{uv}$ and $P_{vw}$ so that $P_{uv}$ is an $r$-qualifying path certifying that $v \in \reach_i(u, G_\sigma)$ and $P_{vw}$ is a weak $r$-qualifying path certifying that $w \in \wreach_{r-i}(v, G_\sigma)$.

\begin{lemma}{\cite{dvorak2013constant}} \label{lem:wcol_bound}
	Let $G = (V, E)$ be a graph, and let $\sigma$ be a total order.
	If for some value of $k$, $\sum_{v \in \reach_r(u, G_\sigma)} (k-1)^{r-d_v} \leq k \cdot (k-1)^{r-1}$ for all $u \in V$, then $|\wreach_r(u, G_\sigma)| \leq \frac{k^{r+1} - 1}{k - 1} - 1$ for all $u \in V$.
\end{lemma}

\begin{proof}
	We proceed by induction on $r$.
	When $r = 1$, the claim holds since $\wreach_1(u, G_\sigma) = \reach_1(u, G_\sigma)$ for all $u \in V$.
	Assume that $r > 1$ and that the claim holds for all smaller values of $r$.
	By Observation~\ref{obs:wreach},
	\begin{align*}
		|\wreach_r(u, G_\sigma)| &\leq \sum_{v \in \reach_r(u, G_\sigma)} 1 + |\wreach_{r-d_v}(v, G_\sigma)|.
		\intertext{By the inductive hypothesis, we can rewrite the sum as}
		|\wreach_r(u, G_\sigma)| &\leq \sum_{v \in \reach_r(u, G_\sigma)} \frac{k^{r-d_v+1} - 1}{k - 1} \\
		&= \sum_{v \in \reach_r(u, G_\sigma)} (k-1)^{r-d_v} \cdot \frac{k^{r-d_v+1} - 1}{(k - 1)^{r-d_v+1}} \\
		&\leq \frac{k^r - 1}{(k-1)^r} \cdot \sum_{v \in \reach_r(u, G_\sigma)} (k-1)^{r-d_v}.
		\intertext{By Lemma~\ref{lem:col_bound}, we can bound the sum, giving}
		|\wreach_r(u, G_\sigma)| &\leq \frac{k^r - 1}{(k-1)^r} \cdot k \cdot (k-1)^{r-1} \\
		&= \frac{k^{r+1} - k}{k-1} \\
		&= \frac{k^{r+1} - 1 - (k-1)}{k-1} \\
		&= \frac{k^{r+1} - 1}{k-1} - 1.
	\end{align*}
\end{proof}

\begin{theorem} \label{thm:alg_main}
	Given a graph $G = (V, E)$ with $\adm_r(G) = k$, Algorithm~\ref{alg:main} produces a total order $\sigma$ such that $\col_r(G_\sigma) \leq k \cdot (k-1)^{r-1}$ and $\wcol_r(G_\sigma) \leq \frac{k^{r+1} - 1}{k-1}$.
\end{theorem}

\begin{proof}
	By Lemma~\ref{lem:bc_noninc}, there is always an unordered vertex with estimated $r$-backconnectivity at most $k$ for any prefix order $\sigma$ that the algorithm produces.
	Then, Lemma~\ref{lem:col_bound} shows that $|\reach_r(u, G_\sigma)| \leq k \cdot (k-1)^{r-1}$ for all $u \in V$, and Lemma~\ref{lem:wcol_bound} extends this to show $|\wreach_r(u, G_\sigma)| \leq \frac{k^{r+1} - 1}{k-1}$ for all $u \in V$.
	Thus, $\col_r(G_\sigma) \leq k \cdot (k-1)^{r-1}$ and $\wcol_r(G_\sigma) \leq \frac{k^{r+1} - 1}{k-1}$.
\end{proof}

\section{Conclusion}
In this paper, we considered the hardness and approximability of the generalized coloring numbers.
These parameters characterize bounded expansion classes and have recently been the basis for new FPT algorithms.
We prove that the $r$-coloring number and the weak $r$-coloring number are both para-NP-hard to compute for $r \geq 2$ when parameterized by the natural parameter $k$.
This implies that neither of the generalized coloring numbers admit an XP algorithm parameterized by $k$.
Moreover, our results imply that there exists a constant $c$ such that it is NP-hard to approximate the generalized coloring numbers within a factor of $c$.
These results do not extend to the related parameters treewidth and treedepth (equivalent to $\col_n(G)$ and $\wcol_n(G)$ respectively) since we explicitly subdivide edges up to $r-1$ times.
We also note that our constructions require the generalized coloring numbers to be at least 5, leaving a few open cases.
Graphs with $r$-coloring number 1 or 2 can be easily recognized (independent of $r$) since this is equivalent to recognizing trees and series-parallel graphs respectively.
However, recognizing graphs with $r$-coloring number 3 or 4 remains open.
This is an inherent limitation of \tcts and \rsat, and will require a different technique to resolve.

We also give an improved approximation for the generalized coloring numbers which runs in time $O(b k n^2)$ where $b$ is the $r$-coloring number of the order produced by the algorithm.
The algorithm greedily selects vertices with small estimated $r$-backconnectivity and gives a $(k-1)^{r-1}$-approximation for the $r$-coloring number and an $O(k^{r-1})$-approximation for the weak $r$-coloring number.
Whether constant factor approximations exist for the generalized coloring numbers remains open, even for $r = 2$.

\clearpage
\bibliographystyle{abbrv}
\bibliography{references}

\begin{thebibliography}{10}

\bibitem{drange2016kernelization}
P.~G. Drange, M.~Dregi, F.~V. Fomin, S.~Kreutzer, D.~Lokshtanov, M.~Pilipczuk,
  M.~Pilipczuk, F.~Reidl, F.~S. Villaamil, S.~Saurabh, S.~Siebertz, and
  S.~Sikdar.
\newblock Kernelization and sparseness: the case of dominating set.
\newblock In {\em 33rd Symposium on Theoretical Aspects of Computer Science
  (STACS 2016)}, volume~47, 2016.

\bibitem{drange2021harmless}
P.~G. Drange, I.~Muzi, and F.~Reidl.
\newblock Kernelization and hardness of harmless sets in sparse classes.
\newblock {\em CoRR}, abs/2111.11834, 2021.

\bibitem{dvorak2013constant}
Z.~Dvo\v{r}\'{a}k.
\newblock Constant-factor approximation of the domination number in sparse
  graphs.
\newblock {\em European Journal of Combinatorics}, 34:833--840, 2013.

\bibitem{dvorak2010deciding}
Z.~Dvořák, D.~Král, and R.~Thomas.
\newblock Deciding first-order properties for sparse graphs.
\newblock In {\em 2010 IEEE 51st Annual Symposium on Foundations of Computer
  Science}, 2010.

\bibitem{grohe2018coloring}
M.~Grohe, S.~Kreutzer, R.~Rabinovich, S.~Siebertz, and K.~Stavropoulos.
\newblock Coloring and covering nowhere dense graphs.
\newblock {\em SIAM Journal on Discrete Mathematics}, 32:2467--2481, 2018.

\bibitem{harpeled2015approximation}
S.~Har-Peled and K.~Quanrud.
\newblock Approximation algorithms for polynomial-expansion and low-density
  graphs.
\newblock In {\em Algorithms - ESA 2015}, 2015.

\bibitem{itai1982complexity}
A.~Itai, Y.~Perl, and Y.~Shiloach.
\newblock The complexity of finding maximum disjoint paths with length
  constraints.
\newblock {\em Networks}, 12:277--286, 1982.

\bibitem{karp1972reducibility}
R.~M. Karp.
\newblock Reducibility among combinatorial problems.
\newblock In {\em Complexity of Computer Computations}, pages 85--103. Springer
  US, 1972.

\bibitem{kierstead2003orderings}
H.~A. Kierstead and D.~Yang.
\newblock Orderings on graphs and game coloring number.
\newblock {\em Order}, 20:255--264, 2003.

\bibitem{kreutzer2019algorithmic}
S.~Kreutzer, I.~Muzi, P.~O. de~Mendez, R.~Rabinovich, and S.~Siebertz.
\newblock {Algorithmic Properties of Sparse Digraphs}.
\newblock In {\em 36th International Symposium on Theoretical Aspects of
  Computer Science (STACS 2019)}, volume 126, 2019.

\bibitem{nadara2019empirical}
W.~Nadara, M.~Pilipczuk, R.~Rabinovich, F.~Reidl, and S.~Siebertz.
\newblock Empirical evaluation of approximation algorithms for generalized
  graph coloring and uniform quasi-wideness.
\newblock {\em ACM Journal of Experimental Algorithmics}, 24:1--34, 2019.

\bibitem{nesetril2008GradAC}
J.~Nesetril and P.~O. de~Mendez.
\newblock Grad and classes with bounded expansion ii. algorithmic aspects.
\newblock {\em European Journal of Combinatorics}, 2008.

\bibitem{nesetril2012sparsity}
J.~Ne\v{s}et\v{r}il and P.~O. de~Mendez.
\newblock {\em Sparsity}.
\newblock Algorithms and Combinatorics. Springer Berlin Heidelberg, 2012.

\bibitem{nesetril2008grad}
J.~Nešetřil and P.~{Ossona de Mendez}.
\newblock Grad and classes with bounded expansion i. decompositions.
\newblock {\em European Journal of Combinatorics}, 29(3), 2008.

\bibitem{reidl2023color}
F.~Reidl and B.~D. Sullivan.
\newblock A color-avoiding approach to subgraph counting in bounded expansion
  classes.
\newblock {\em Algorithmica}, 85, 2023.

\bibitem{reidl2019neighborhood}
F.~Reidl, F.~S. Villaamil, and K.~S. Stavropoulos.
\newblock Characterising bounded expansion by neighbourhood complexity.
\newblock {\em European Journal of Combinatorics}, 75, 2019.

\bibitem{tovey1984simplified}
C.~A. Tovey.
\newblock A simplified {NP}-complete satisfiability problem.
\newblock {\em Discrete Applied Mathematics}, 8:85--89, 1984.

\bibitem{zhu2009coloring}
X.~Zhu.
\newblock Colouring graphs with bounded generalized colouring number.
\newblock {\em Discrete Mathematics}, 309:5562--5568, 2009.

\end{thebibliography}

\end{document}